\documentclass[%
aps,
pra,%
amsmath,amssymb,showpacs,
reprint,%
floatfix,
twocolumn,
nofootinbib,
longbibliography,
superscriptaddress
]{revtex4-1}
\usepackage[utf8]{inputenc} 
\usepackage[T1]{fontenc}
\usepackage{graphicx}
\usepackage{dcolumn}
\usepackage{bm}
\usepackage[justification=justified, format=plain]{caption}
\usepackage{subcaption}
\usepackage[normalem]{ulem}
\usepackage{physics}
\usepackage{mathtools}
\usepackage{amsmath,amsfonts,amssymb}
\usepackage{amsthm}
\usepackage[dvipsnames]{xcolor}
\usepackage{verbatim}
\usepackage{bbold}

\DeclareMathAlphabet\mathbfcal{OMS}{cmsy}{b}{n}

\usepackage[unicode=true,bookmarks=false,breaklinks=false,pdfborder={0 0 1},colorlinks=true]{hyperref}
\hypersetup{
 citecolor=blue,linkcolor=blue,urlcolor=blue}

\newtheorem*{lemma}{Lemma}
\theoremstyle{remark}
\definecolor{LightBlue}{rgb}{0,0.65,0.8}

\begin{document}

\title{Temporal modes of quantum states of light scattered by a two-level system}
\author{Yann Bouchereau}
\thanks{These authors contributed equally to the writing of this work.}
\author{Lucas Weitzel}
\thanks{These authors contributed equally to the writing of this work.}
\affiliation{Laboratoire Kastler Brossel, Sorbonne Université, CNRS, ENS-Université PSL, Collège de France, 4 place Jussieu, Paris F-75252, France}
\author{Valérian Thiel}
\affiliation{Pasqal, 24 Av. Emile Baudot, 91120 Palaiseau, France}
\author{Valentina Parigi}
\author{Mattia Walschaers}
\author{Nicolas Treps}
\affiliation{Laboratoire Kastler Brossel, Sorbonne Université, CNRS, ENS-Université PSL, Collège de France, 4 place Jussieu, Paris F-75252, France}

\date{\today}

\begin{abstract}
    The deterministic generation of non-Gaussian quantum states of light, of importance to quantum computing, is a challenging problem due to the difficulty to control nonlinearities in physical systems. In this work, we characterize the light stemming from one of the most fundamental quantum optics configurations: the unidirectional scattering of multimode and multiphoton light by a two-level system. We provide an analytic and explicit description of the output light solely in terms of the corresponding input temporal modes which allows a straightforward physical interpretation and is computationally more effective compared to numerical methods. Then, we focus on the specific case of the scattering of two photons in a single mode. By numerically decomposing the output state in terms of its principal modes, we find that it is possible to map single-mode two-photon inputs to be into two-mode entangled output states, i.e., two-photon NOON states, to very good approximation. The latter states, in turn, are known to have more Wigner negativity compared to the associated input, which ultimately suggests a potential application of our considered setup in the deterministic generation of non-Gaussian states.
\end{abstract}

\maketitle

\section{Introduction}

Non-Gaussian states of light are of fundamental importance in modern quantum optics, with significant implications for quantum computing. Nevertheless, their generation is challenging, and efforts in this direction have pushed developments in a multitude of platforms for quantum technologies. Commonly employed methods are photon addition and subtraction \cite{doi:10.1126/science.1103190,doi:10.1126/science.1122858, doi:10.1126/science.1146204}, breeding protocols \cite{PhysRevLett.114.193602}, and photon catalysis \cite{Lvovsky2002,Eaton_2019,ktc9-9rjb}. Importantly, these techniques effectively use measurements to implement nonlinear processes, which are necessary prerequisites for the generation of non-Gaussian states \cite{PRXQuantum.2.030204}.

Despite the remarkable results achieved by the aforementioned approaches, low success probabilities remain a major bottleneck for large-scale implementations. In other words, they cannot \textit{deterministically} generate a desired non-Gaussian state, precisely because their protocols are intrinsically probabilistic. This observation continues to fuel interest in alternative, deterministic methods for implementing non-Gaussian operations.

To that end, promising platforms are superconducting and trapped ions setups, in which the interplay between bosonic modes and a two-level system generally gives access to controllable nonlinearities \cite{Chang2014}. In the optical regime, this setting is harder to realize, but it has been demonstrated with Rydberg atoms in a cavity \cite{Rempe19,Margo23}, or by placing quantum emitters in the vicinity of waveguides \cite{Goban2014,PhysRevLett.117.133603}. While the basic physics of this setting can be captured through the Jaynes-Cummings model \cite{Jaynes1963}, the theory rapidly grows more complicated when we add more degrees of freedom. In particular, for a setting where the incoming light is no longer considered to be monochromatic, we are inevitably confronted with a significantly more complex multimode problem.

On the theory side, several recent works have focused on a paradigmatic nonlinear problem: The one-dimensional scattering of multimode light by a two-level system. Several models have been developed to study the spectral properties of quantum light in this setting. On the one hand, the method of \cite{Kiilerich2019, Kiilerich2020} introduces virtual cavities in combination with the input-output formalism. This approach is scalable, versatile \cite{lim2026optimizingwignernegativityscattering,copie2026DeterministicPrep}, and has found applications in circuit- \cite{Yang2025} and cavity-QED \cite{Elliott2024}, but it requires the use of purely numerical techniques. On the other hand, the collisional model of \cite{Maffei2022} has the advantage of being analytically tractable, but its recursive relations make it difficult to implement for certain classes of modes. 

In this work, we present an alternative, fully analytical treatment of the aforementioned problem. We use tools from multimode quantum optics \cite{fabreModesStatesQuantum2020}, which give a compact analytical solution to this problem when combined with the formalism of scattering theory. This allows us to perform systematic studies over the input modes and identify the relevant modes in the output. To make our analytical approach tractable, we develop the theory in the Fock basis to obtain expressions for an arbitrary (finite) number of photons. We illustrate the usefulness of our approach by conducting a case study of two-photon states that occupy different types of spectral modes.

This paper is organized as follows. In Sec.~\ref{sec:theoretical_background}, we present our physical setup of interest and briefly discuss the relevant concepts for the development of our approach. In Sec.~\ref{sec:output_state}, we introduce one of the main results of this work, which is an intuitive and compact way of describing a multiphoton and multimode quantum state of light scattered by a two-level system in a waveguide. In Sec.~\ref{sec:modal}, we present the mathematical tools needed to analyze these scattered states. Next, in Sec.~\ref{sec:two-photon}, our method is employed in the case of the scattering of two photons, for the analysis of the modal content of the photonic output state. Finally, in Sec.~\ref{sec:summary}, we conclude by summarizing our findings and give prospects for future applications of our approach.

\section{Theoretical background}
\label{sec:theoretical_background}

\subsection{Setup}
\subsubsection{Hamiltonian and temporal modes}
\label{sec:hamiltonian_setup}

The setup of interest consists in a single two-level system (TLS) coupled to a one-dimensional continuum of photonic modes. The dynamics of the whole system is governed by the Hamiltonian $H_{\rm tot} = H_{\rm ph} + H_{\rm at} + H_{\rm int}$, where
\begin{align}
    \label{H_ph}
    H_{\rm ph} &= \int_{-\infty}^{\infty} \dd k\;\omega_k a^\dag(k) a(k),\\
    \label{H_at}
    H_{\rm at} &= \omega_{\rm ge} \sigma_{\rm ee},\\
    \label{H_int}
    H_{\rm int} &= \sqrt{\frac{\Gamma}{4\pi}}\int_{-\infty}^\infty \dd k\;a(k) \sigma_{\rm eg} + \text{h.c.},
\end{align}
are, respectively, the photonic field, the two-level system, and the light-matter interaction Hamiltonians. Eqs.~(\ref{H_ph}--\ref{H_int}) are written in units such that $\hbar =c=1$, and we assume the dispersion relation to be $\omega_k=|k|$ for the photon frequency associated with momentum $k$. Furthermore, $\omega_{\rm ge}$ denotes the atomic transition frequency between the excited ($\ket{\rm e}$) and ground ($\ket{\rm g}$) states, $\sigma_{\rm ee}=\ketbra{\rm e}$, $\sigma_{\rm eg}=\ketbra{\rm e}{\rm g}$ and $\Gamma$ is the spontaneous emission rate of the TLS.

The ladder operators $a^\dagger(k)$, $a(k)$, employed in Eqs.~\eqref{H_ph} and~\eqref{H_int}, represent, respectively, the creation and the annihilation of a photon with wavenumber $k$. They obey the standard commutation relation
\begin{equation}
    \label{eq:commutation_momentum}
    [a(k), a^\dag(p)] = \delta(k-p).
\end{equation}

In this paper we will mostly work in the basis of ``time-localized'' photons. The associated annihilation operator is defined by the Fourier transform with the following convention:
\begin{equation}
    \label{eq:Fourier}
    a(t) \coloneq \int_{\mathbb{R}} \dd k\; e^{-i(k-\omega_{\rm ge})t}a(k),
\end{equation}
and obey a (time-domain) commutation relation analogous to Eq.~\eqref{eq:commutation_momentum}:
\begin{equation}
    \label{eq:commutation_time}
    [a(t), a^\dag(u)] = \delta(t-u),
\end{equation}
where, in the above, $t$ and $u$ are time variables. $a^\dagger(t)$ is thus the creation operator of a photon at the instant $t$.

Finally, it will also be useful to define ladder operators for a \textit{temporal mode} $f$ as \cite{fabreModesStatesQuantum2020}
\begin{align}
    \label{eq:mode_creation}
    a^\dag_{f} \coloneq \int_{\mathbb R} \dd u f(u) a^\dag(u).
\end{align}
Therefore, $a^\dag_{f}$ creates a photon in mode $f$ --- instead of at an instant $u$ -- which, in turn, can be any square-integrable normalized function (i.e. $\int_{-\infty}^{\infty}\dd u|f(u)|^2 = 1$). The commutation relation for the temporal mode ladder operators is simply given by the scalar product between the underlying temporal modes:
\begin{align}
    \label{eq:commutator_mode_creation}
    [a^{\strut}_{f}, a^\dag_{g}] = \int_{\mathbb R} \dd u f^*(u) g(u),
\end{align}
as it can be shown directly from Eq.~\eqref{eq:commutation_time}.

\subsubsection{Input and output states}
\label{sec:in_out_states}

We consider an incident $n$-photon, multimode, pure state of the (free) field, which propagates unidirectionally. This state is scattered by the TLS, whose initial and final states are both assumed to be the ground state. This setup is illustrated in Fig.~\ref{fig:setup} and can be thought of as describing, for example, an atom coupled to a chiral waveguide\footnote{We point out that the chiral waveguide goes beyond being just a toy model, and does correspond to realizable physical systems \cite{Haldane2008, Wang2008}. Furthermore, the forward and backward propagating solutions of the chiral system can be used to construct solutions for the associated non-chiral system \cite{Shen2015}.} \cite{Caneva2015}.

The input state of the field is thus written as \cite{fabreModesStatesQuantum2020}
\begin{align}
    \label{eq:input_state}
    \ket{\psi_{\rm in}} &=\mathcal N a^\dagger_{f_n}\ldots a^\dagger_{f_1}\ket{0}=\mathcal{N}\prod_{j=1}^n\left[\int_{\mathbb R}\dd u f_j(u)a^\dagger(u)\right]\ket{0},
\end{align} 
We emphasize the fact that, here and in the following, there are no additional assumptions about the modes $\{f_j\}$ apart from being square-integrable. Notably they can be identical (corresponding to a state with $n$-photons in a single temporal mode), orthogonal or have a non-zero overlap. The normalization factor $\mathcal N$ ensures that $|\!\braket{\psi_{\rm in}}\!|^2=1$. An explicit expression for $\mathcal N$ as a function of the modes $(f_1,\ldots,f_n)$ is given in Appendix~\ref{appendix:normalization}.

In the present work, we are interested in the representation of the free field's output state, that is, long after being scattered by the TLS. This state can be written as
\begin{align}
    \label{eq:output_state}
    \ket{\psi_{\rm out}} &= \frac{1}{\sqrt{n!}}\int_{\mathbb R^n} \dd^n\vb t\; \Psi_{\rm out}(\vb t)a^\dagger(t_n)\ldots a^\dagger(t_1)\ket{0},
\end{align}
where here and in the following the boldface letters denote an array of variables, i.e., $\vb t = (t_1,\ldots,t_n)$ for the above example, and the output state wavefunction is defined as
\begin{align}
    \label{eq:output_wave_definition}
    \Psi_{\rm out}(\vb t)\coloneq\frac 1{\sqrt{n!}}\bra{0}\prod_{j=1}^n a(t_j)\ket{\psi_{\rm out}},
\end{align}
Eqs.~\eqref{eq:output_state} and~\eqref{eq:output_wave_definition} are normalized such that $|\!\braket{\psi_{\rm out}}\!|^2=1$ and $\int_{\mathbb R^n} \dd^n\vb t |\Psi_{\rm out}(\vb t)|^2=1$ (for the derivation, see Appendix \ref{appendix:normalization}).
\begin{figure*}[t]
    \centering
    \includegraphics[trim={0cm 9cm 0 5cm},clip, width=0.7\linewidth]{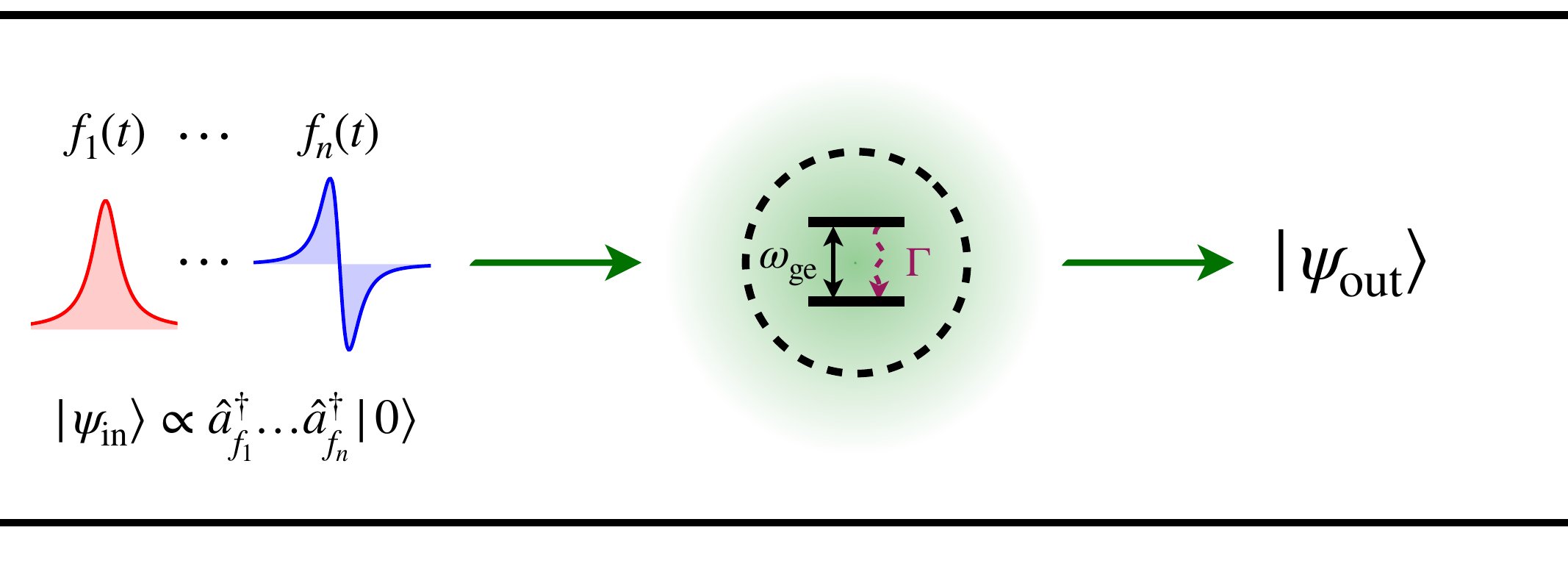}
    \caption{The setup of the considered physical system. A two-level system (TLS), with transition frequency $\omega_{\rm ge}$ and decay rate $\Gamma$, is placed within a chiral waveguide. A multiphoton incoming state of light $\ket{\psi_{\rm in}}$, prepared in the temporal modes $f_1(t),\ldots,f_n(t)$, interacts with the TLS --- here indicated by the green shade --- and evolves to the asymptotic output state $\ket{\psi_{\rm out}}$. The horizontal dimension of the figure should be interpreted as the spatial domain.}
    \label{fig:setup}
\end{figure*}

Our goal is to give an explicit expression of the output state and to provide a ``modal'' description of it. We give a more precise meaning of this idea in Sec.~\ref{sec:modal}. In order to investigate the description of scattered states, the scattering matrix is a natural tool. We therefore discuss its properties in the next subsection.

\subsection{Scattering matrix}
\label{sec:scattering_matrix}

\subsubsection{Definition and relation to the output state}

The scattering matrix (or $S$-matrix) is an operator $\mathcal S$ that maps an input state to a scattered output state, which, in our case, are given by Eqs.~\eqref{eq:input_state} and~\eqref{eq:output_state}, respectively. Mathematically, this relation reads \cite{Taylor2012}
\begin{align}
    \label{eq:application_scatter}
    \ket{\psi_{\rm out}}=\mathcal S\ket{\psi_{\rm in}},
\end{align}
and we define the corresponding matrix elements associated with the scattering of $n$ photons by
\begin{align}
    \label{eq:s_matrix_time}
    S^{(n)}_{(\vb t,\vb u)}\coloneq \bra{0}a(t_1)\ldots a(t_n)\mathcal Sa^\dagger(u_1)\ldots a^\dagger(u_n)\ket{0}.
\end{align}

From Eq.~\eqref{eq:application_scatter} and using Eqs.~\eqref{eq:input_state},~\eqref{eq:output_wave_definition} and~\eqref{eq:s_matrix_time}, we show in Appendix~\ref{appendix:normalization} that the output state wavefunction can be written in terms of the input temporal modes $\{f_j\}$ as
\begin{align}
    \label{eq:output_wavefunction}
    \Psi_{\rm out}(\vb t)=\frac{\mathcal{N}}{\sqrt{n!}}\int_{\mathbb R^n} \dd^n\vb u\; \left[\prod_{j=1}^n{f_j(u_j)}\right] S^{(n)}_{(\vb t,\vb u)}.  
\end{align}

As evidenced by Eq.~\eqref{eq:output_wavefunction}, we see that the scattering matrix plays a crucial role in the computation of $\Psi_{\rm out}(\vb t)$. For that reason, it will be relevant to introduce an explicit expression for $S^{(n)}_{(\vb{t},\vb{u})}$ which is already available in the literature, and discuss its physical interpretations.

\subsubsection{Analytical expression of $S^{(n)}$}

In \cite{Caneva2015}, the authors derive an analytical expression of the $S$-matrix coefficients, as defined in Eq.~\eqref{eq:s_matrix_time}, in time (and frequency) domain. We reproduce here the basic idea of this derivation. To that end, we observe that Eq.~\eqref{eq:s_matrix_time} implies the symmetry of $S^{(n)}_{(t_1,\ldots,t_n), (u_1,\ldots,u_n)}$ under any permutation of the indices $(t_1,\ldots,t_n)$ and $(u_1,\ldots,u_n)$, respectively. Therefore, it suffices to provide, in the following, the explicit expression for $S^{(n)}_{\vb t, \vb u}$ --- and similarly for $\Psi_{\rm out}(\vb t)$ --- with $\vb t$ and $\vb u$ in ascending order (i.e. $t_1\leq t_{2} \leq \ldots \leq t_n$), which leads to more compact expressions.

First, by Fourier-transforming and integrating the evolution equations of $a(k)$ under the total Hamiltonian $H_{\rm tot}$, we get the input-output equation \cite{Gardiner1985, Wallis1995, Fan2010input}
\begin{align}
    \label{eq:input_output}
    a_{\rm out}(t)=a_{\rm in}(t)-i\sqrt{\Gamma}\sigma_{\rm ge}(t).
\end{align}
The input operator $a_{\rm in}(t)$ is here equivalent to $a(t)$ just before the scattering event, while $a_{\rm out}(t)$ gives the evolution of $a_{\rm in}(t)$ right after the interaction.\footnote{Formally, $a^\dagger_{\rm in}(t)$ and $a^\dagger_{\rm out}(t)$ create scattering eigenstates of the \textit{total} Hamiltonian $H_{\rm tot}$, while $a^\dagger(t)$ creates an eigenstate of the (free) \textit{photonic} Hamiltonian $H_{\rm ph}$ [see Eq.~\eqref{H_ph}]. For a more in-depth discussion, the authors refer to \cite{Gardiner1985, Wallis1995, Fan2010input}.} The $S$-matrix elements can be subsequently expressed in terms of the input and output operators as \cite{Fan2010input}
 \begin{align}
    \label{eq:s_matrix_in_out}
    S^{(n)}_{(\vb{t},\vb{u})} = \bra{0}a^{\mathstrut}_{\rm out}(t_1)\ldots a^{\mathstrut}_{\rm out}(t_n) a_{\rm in}^\dag(u_1)\ldots a_{\rm in}^\dag(u_n) \ket{0}.
\end{align}

The next step is to plug \eqref{eq:input_output} into \eqref{eq:s_matrix_in_out} and to use the commutation relations satisfied by $a_{\rm in}, a_{\rm out}$. After a lengthy derivation, available in \cite{Caneva2015}, one finds that $S^{(n)}_{\vb t, \vb u}$ can be expressed only in terms of atomic operators as
\begin{align}
    \label{eq:scattering_matrix}
    S_{(\vb t,\vb u)}^{(n)}&=\sum_{m=0}^n  \sum_{I \in\mathcal P(m)} \sum_{J \in\mathcal P(m)}  i \mathcal{T}^{(m)}_{\vb t[I],\vb u[J]}\times \prod_{l =1}^{n-m} \delta(t_{\Bar I_l} - u_{\Bar J_l}),
\end{align}
where $\mathcal{P}(m)$ denotes the set of all subsets of $\{1,\ldots,n\}$ containing $m$ elements, and, for $I\in\mathcal{P}(m)$, $\Bar I = \{1,\ldots,n\} \setminus I$ is the complementary of $I$ in $\{1,\ldots,n\}$. With this, $\vb t[I]=(t_{I_1},\ldots,t_{I_m})$ denotes the array containing a certain subset of the time variables $\vb t$. Finally, $ i \mathcal{T}^{(m)}_{\vb t[I],\vb u[J]}$ is the time-ordered atomic correlation function of order $m$, defined as
\begin{align}
    \label{eq:connected_equation}
    i &\mathcal{T}^{(m)}_{\vb t[I],\vb u[J]}\nonumber\\
    &= \bra{0} T\left[\sigma_{\rm ge}(t_{I_1})\ldots \sigma_{\rm ge}(t_{I_m})\sigma_{\rm eg}(u_{J_1})\ldots \sigma_{\rm eg}(u_{J_m})\right]\ket{0},
\end{align}
where $T[\;\cdot\;]$ denotes time-ordering. In the case of $n$ photons coupled to a single TLS, Eq.~\eqref{eq:connected_equation} can be calculated explicitly: $ i \mathcal{T}^{(m)}_{\vb t[I],\vb u[J]}=1$ for $m=0$, and, otherwise,
\begin{widetext}
\begin{align}
    \label{eq:connected_s}
    i\mathcal T^{(m)}_{\vb t[I],\vb u[J]} = \left(-\Gamma\right)^m\left[\prod_{l=1}^{m}e^{-\Gamma(t_{I_l}-u_{J_l})/2}\right] \Theta\left(t_{I_m} >u_{J_m}>\ldots >t_{I_1}>u_{J_1}\right),
\end{align}
\end{widetext}
where we defined 
\begin{align}
    \Theta(x_n>x_{n-1}>...>x_1)=
    \begin{cases}
        1,\;\;\text{if}\;\;x_n>x_{n-1}>...>x_1,\\
        0,\;\;\text{otherwise}.
    \end{cases}
\end{align}

Physically, the delta factors in $S_{(\vb t,\vb u)}^{(n)}$ are associated with photons that are scattered without having interacted with the TLS. On the other hand, the $i\mathcal T^{(m)}_{\vb t[I],\vb u[J]}$ factors --- also referred to as the connected part of the scattering matrix \cite{Caneva2015, Xu2015, Shi2009} --- encompass the events where $m$ photons have interacted with the TLS and thus can have exchanged their momenta.

By exploiting Eqs.~\eqref{eq:scattering_matrix} and~\eqref{eq:connected_s}, we provide, in the next section, a computationally effective and conceptually intuitive expression for the general $n$-photon output wavefunction in the time domain, as defined in Eq.~\eqref{eq:output_wavefunction}. 

\section{Multiphoton output state in the time domain}
\label{sec:output_state}

In this section, we apply the tools presented previously to compute the $n$-photon output state wavefunction scattered by the TLS in the chiral waveguide. In order to gain intuition in the general case, we start by discussing the case of the single-photon scattering.

\subsection{Single photon case}

The single-photon scattering matrix in the time domain representation is obtained from Eqs.~\eqref{eq:scattering_matrix} and~\eqref{eq:connected_s} by setting $n=1$:\footnote{$S^{(1)}_{(t,u)}$ is equivalently recovered by applying the Fourier transform, as defined in Eq.~\eqref{eq:Fourier}, to the more familiar representation of the single-photon scattering matrix in momentum domain \cite{Wallis1995, Fan2010input, Xu2015}: $\tilde S^{(1)}_{(p,k)}=\delta(k-p)(k-i\Gamma/2)/(k+i\Gamma/2)$.}
\begin{align}
    S^{(1)}_{(t,u)} =\delta(t-u)-\Gamma e^{-\Gamma(t-u)/2}\Theta(t>u).
\end{align}
Therefore, using Eq.~\eqref{eq:output_wavefunction}, the scattered state's wavefunction of a single photon input in mode $f(t)$ is given by
\begin{align}
    \label{eq:single-photon-time}
    \Psi_{\rm out}(t)&=\int_{-\infty}^{\infty}\dd u\;f(u)S^{(1)}_{(t,u)}=f(t)+F(t,-\infty),
\end{align}
where we introduced
\begin{align}
    \label{eq:convoluted_mode}
     F(t,u) =-\Gamma \int_{u}^{t}\dd\tau\;e^{-\Gamma(t-\tau)/2} f(\tau).
\end{align}

In Eq.~\eqref{eq:single-photon-time}, we observe that the output state can be separated into two contributions: a non-interacting and a interacting part, $f(t)$ and $F(t,-\infty)$, respectively. Physically, the former is simply the same as the input mode, while the latter describes the contribution to the output state coming from the re-emission of the photon by the atom. This feature is compatible with results in the literature that have been derived via different methods, and for related setups \cite{Domokos2002, Roulet2016, Maffei2022, Dkabrowska2022}.

\subsection{Multiphoton case}
\label{sec:multiphoton}

Let us now turn to the scattering of $n$ photons. Again, the scattered state wavefunction is obtained by plugging Eq.~\eqref{eq:scattering_matrix} into Eq.~\eqref{eq:output_wavefunction}. This calculation is performed in detail in Appendix~\ref{appendix:derivation_general_formula}, where we consider the general case of $n$ different modes. In this subsection, though, we stick to the simpler case where all $n$ photons are in the same temporal mode $f$. The main result of this work is that the output wavefunction can be expressed as
\begin{align}
    \label{eq:output_n_photons}
    &\Psi_{\rm out}(\vb t)=\sum_{m=0}^n\sum_{I\in \mathcal P(m)}\prod^{m}_{j=1} F(t_{I_j},t_{I_{j-1}})\prod_{l=1}^{n-m} f(t_{\Bar I_l}),
\end{align}
with the convention $t_{I_0} \to -\infty$. As in Eqs.~\eqref{eq:scattering_matrix} and~\eqref{eq:connected_s}, the time ordering of the $\vb t[I]$ variables is implicit. The generalization of Eq.~\eqref{eq:output_n_photons} for different modes is straightforward, with the additional step of averaging over all permutations of the distinct input modes [for the derivation, see Appendix~\ref{appendix:derivation_general_formula}].

Despite its derivation not being immediate, Eq.~\eqref{eq:output_n_photons} has a straightforward physical interpretation. Indeed, each of the $n$ photons contributing to the output state may or may not have interacted with the TLS. If we choose $m$ the number of photons that interacted and $(I_1,\ldots,I_m)$ their indices, then the photons that did not interact are left unchanged, giving the $n-m$ factors $f(t_{\Bar I_l})$, while the $m$ photons that interacted give the factors of the form $F(t_{I_m},t_{I_{m-1}})$. The full output wavefunction is then simply given by the sum over all possible choices for $m$ and $I$. Finally $F(t_{I_m},t_{I_{m-1}})$ [Eq.~\eqref{eq:convoluted_mode}] is a convolution of the input temporal mode $f$ with the decaying exponential of the atomic response over a finite time interval, and thus corresponds to an interaction with the TLS between times $t_{I_{m-1}}$ and $t_{I_{m}}$.

With such an interpretation, Eq.~\eqref{eq:output_n_photons} explicitly demonstrates a peculiar feature of the scattering of $n$ photons by a TLS: the output wavefunction contains no term featuring simultaneous interaction of two or more photons with the TLS. Instead, only \textit{sequential} scattering of photons is allowed to happen.\footnote{Sequential photon scattering can also be understood as arising from destructive interference between the coherent and incoherent photon-scattering contributions to the output state that are associated to two or more photons \cite{Dalibard1983, Masters2023}.} This purely quantum phenomenon, sometimes referred to as photon anti-bunching, is, in fact, well-known in the literature: it has been theoretically deduced \cite{Stoler1974, Carmichael1976} and experimentally observed \cite{Kimble1977}, by the analysis of the photonic temporal second-order correlation function.

A detailed description of a multimode and multiphoton scattered state of light, however, was only investigated more recently, for setups slightly different from the one considered in this work, and with alternative approaches\footnote{In \cite{Roulet2016, Maffei2022}, the intermediate dynamics of the field and the atom were taken into account into the formalism. In our case, in contrast, the atomic dynamics is not calculated, and only the field's (time-asymptotically) initial and final states matter.} \cite{Roulet2016, PhysRevA.93.063807, Maffei2022}. Nevertheless, we believe our result presents the most compact way of describing the field's output state. Moreover, Eq.~\eqref{eq:output_n_photons} significantly reduces the complexity of calculating the former by bringing the number of integrals to be performed down to one, but evaluated at different times.

Finally, it is worth mentioning that, although Eq.~\eqref{eq:output_n_photons} holds for an input state with a fixed number of photons only, its generalization to any superposition of input states with different number of photons is straightforward. This property is due to the fact that the scattering matrix [Eq.~\eqref{eq:application_scatter}] keeps the photon number invariant. Consequently, a more generic output state can be written as a superposition of individual $n$-photon scattered states of the form of Eq.~\eqref{eq:output_state}. 

\section{Modal decomposition of the output state}
\label{sec:modal}

The formula we obtain for the output state [Eq.~\eqref{eq:output_n_photons}] is compact, analytical, and can be computed efficiently. However, the question of defining and extracting the relevant modes of the problem remains. 
In the present section, we briefly describe the framework of the modal decomposition of the output state, which will be necessary for the discussion of Sec.~\ref{sec:two-photon}. For a more in-depth discussion, we refer to \cite{fabreModesStatesQuantum2020}.

The aim of modal decomposition is to answer the following questions: What is the minimal number of modes necessary to fully describe a given multimode quantum state $\ket{\Psi}$, what are these modes and can we associate them with a direct physical interpretation? In the case of our $n$-photon multimode output state described by the wavefunction $\Psi_{\rm out}(\vb t)$, there always exists decompositions of the form
\begin{equation}
    \Psi_{\rm out}(\vb t) = \sum_{j=1}^{\mathcal M} \lambda_j \psi_{1,j}(t_1)\ldots \psi_{n,j}(t_n),
    \label{eq:gen_modal_decomp}
\end{equation}
for some integer $\mathcal M$, real parameters $\lambda_j$ and (not necessarily orthogonal) temporal modes $\psi_{i,j}$. In practice, the task of modal decomposition becomes finding a decomposition such as Eq 25 where $\mathcal M$ is minimal and the $\psi_{i,j}$ carry physical meaning.
In the general case, this problem is not well-defined and there is not a unique decomposition such as Eq.~\eqref{eq:gen_modal_decomp} with these properties, making modal decomposition a challenging task \cite{ktc9-9rjb,lopeteguigonzalez2025}.

For $n=2$ the problem is much easier, since the output state wavefunction depends on only two time variables and we can employ existing tools from linear algebra in our investigation by considering $\Psi_{\rm out}$ as a matrix. Here $\Psi_{\rm out}$ is complex symmetric, which means that we can obtain a reduction of the form of Eq~\eqref{eq:gen_modal_decomp} by performing the Takagi decomposition \cite{Chebotarev2014}:
\begin{align}
     \label{eq:eigendecomposition}
     \Psi_{\rm out}(t,t')=\sum_{j}\lambda_j\psi_j(t)\psi_j(t'),
\end{align}
where the $\lambda_j$ are real non-negative and the family $\{\psi_j\}$ is complex orthonormal. The temporal modes $\psi_j$ are called \textit{principal modes}.

This output wavefunction decomposition can be equivalently written in Dirac notation. Plugging Eq~\eqref{eq:eigendecomposition} into Eq~\eqref{eq:output_state} we readily arrive at
\begin{align}
    \label{eq:state_eigendecomposition}
    \ket{\psi_{\rm out}}=\frac 1 {\sqrt{2}} \sum_j\lambda_j(a^\dagger_{\psi_j})^2\ket{0}=\sum_j\lambda_j\ket{2_{\psi_j}}.
\end{align}
From Eq.~\eqref{eq:state_eigendecomposition}, we thus see that the principal modes form the basis in which the output state is described as a superposition of two-photon Fock states.

If we consider only the case where the input modes are real, the output wavefunction will be real as well and Eq~\eqref{eq:eigendecomposition} is simply the eigendecomposition of $\Psi_{\rm out}$. In this case,$\{\psi_j\}$ is real and the eigenvalues $\lambda_j$ can be negative, which physically corresponds to a $\pi$ phase-shift in the corresponding mode. 


Finally, an important quantity to be investigated is the effective number of modes, defined as \cite{fabreModesStatesQuantum2020}
\begin{align}
     \label{eq:mode_number}
     M\coloneq \frac{1}{\sum_j\lambda_j^4}.
\end{align}
This quantity gives the intrinsic number of modes that is necessary to describe the output state.

In the following, we will apply the concepts introduced in the present section to investigate the modal content in the scattering of two-photon states.

\section{Case study: two-photon input state}
\label{sec:two-photon}

In this section, we study two-photon input states where both incoming photons are in the same mode $f$. In this case the input state is \textit{single-mode} in the sense of Eq~\eqref{eq:mode_number}, meaning that $M=1$ for all input states in this section.

The two-photon output state wavefunction is hence recovered from Eq.~\eqref{eq:output_n_photons} by setting $n=2$, and reads\footnote{We recall that Eq.~\eqref{eq:output_two_photon} is valid only when $t\geq t'$, which was so far left implicit. To recover the wavefunction in the case when $t'\geq t$, one then has to swap the time arguments in the equation.}
\begin{align}
    \label{eq:output_two_photon}
    \Psi_{\rm out}(t,t')\big|_{t\geq t'}&=f(t) f(t')\nonumber\\
    &+f(t)F(t',-\infty)+f(t')F(t,-\infty)\nonumber\\
    &+F(t,t')F(t',-\infty).
\end{align}
Using Eq.~\eqref{eq:output_two_photon}, the output state decomposition and analysis is performed numerically for two different choices of input modes, as we discuss in the following.

\subsection{Lorentzian mode}
\label{sec:input_mode_analysis}

\subsubsection{Input mode}

In the momentum domain, the Lorentzian mode reads
\begin{align}
    \tilde f_L(p) = \sqrt{\frac\gamma\pi}\frac{\gamma/2}{(\gamma/2)^2+p^2},
\end{align}
where $\gamma$ is the Lorentzian's full width at half maximum. The mode corresponding to $\tilde f_L(p)$ in the time domain is a two-sided decaying exponential:
\begin{align}
    \label{eq:lorentzian_time}
     f_L(t)=\sqrt{\frac{\gamma}2}e^{-\gamma|t|/2}.
\end{align}
This mode shape, apart from leading to an explicit expression for the output state, is also of practical relevance. For example, it arises naturally as the single-photon output mode profile from cavity-based devices, such as, e.g., continuous-wave optical parametric oscillators \cite{Nielsen2007, Morin2013}.

\subsubsection{Output mode analysis}
\label{sec:output_mode_analysis}

By plugging Eq.~\eqref{eq:lorentzian_time} into Eq.~\eqref{eq:output_two_photon}, we obtain the output state's wavefunction for the Lorentzian input. In this case, the computation can even be performed analytically; the explicit expression for the output state wavefunction is given in Appendix~\ref{appendix:explicit_expression_lorentzian}. Then, with the full expression of the output state's wavefunction, we numerically compute the eigenvalues $\lambda_j$ and the effective number of output modes $M$ for different values of the ratio $\gamma/\Gamma$. These results are plotted in Fig.~\ref{fig:number_and_singular_values}.

\begin{figure}
    \centering
    \includegraphics[width=\linewidth]{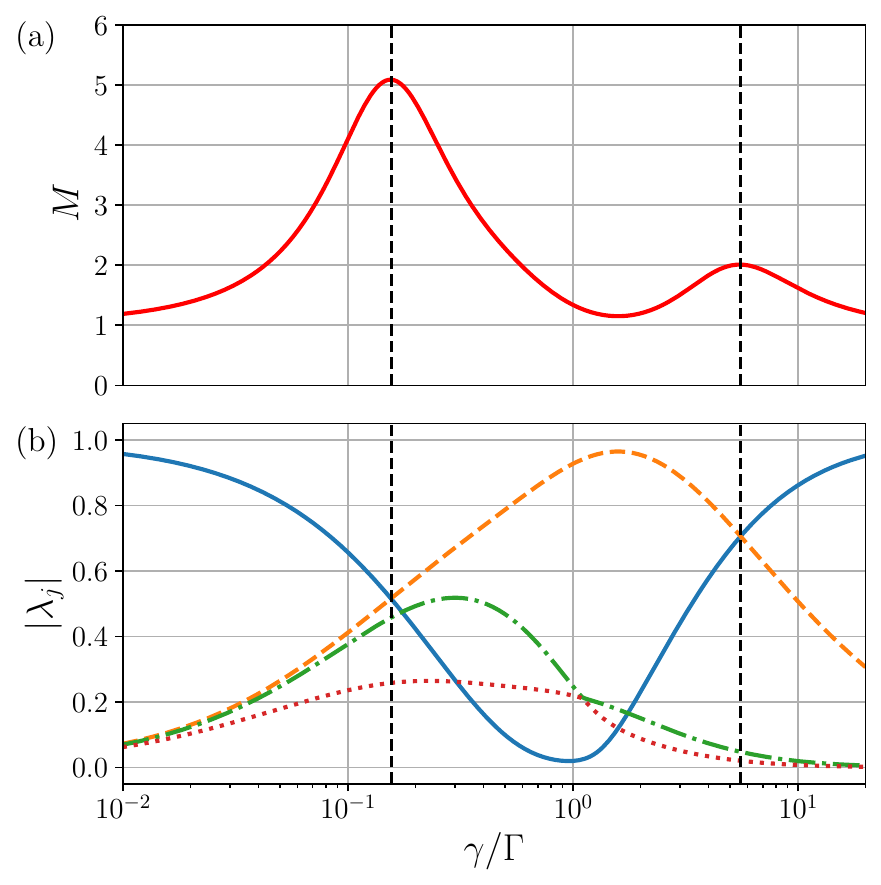}
    \caption{Analysis of the modal content of the output state for the Lorentzian input. (a): Effective number of modes $M$ [see Eq.~\eqref{eq:mode_number}] and (b): Moduli of the eigenvalues $\{\lambda_j\}$ for $j=1,2,3,4$ (in solid, dashed, dashed-dotted, and dotted lines, respectively), as a function of the Lorentzian width $\gamma$ in units of the TLS decay rate $\Gamma$, in logarithmic scale. The only positive eigenvalue is $\lambda_1$. The dashed vertical black lines highlight the local maxima of $M$, $M\approx 5.15$ and $M\approx 2$, which are located at $\gamma/\Gamma\approx0.15$ and  $\gamma/\Gamma\approx5.52$, respectively.}
    \label{fig:number_and_singular_values}
\end{figure}

As a consistency check, we first note that, in the limit $\gamma/\Gamma\to\infty$ --- i.e. when the temporal mode $f_L(t)$ is much narrower than the atomic lifetime --- the average number of modes goes to $1$. This behavior is expected since the interaction time between the input pulse and the TLS goes to zero in this limit, meaning that they do not interact.

In the other extreme, when $\gamma/\Gamma \to  0$, the input mode width is much broader than the atomic lifetime. Mathematically, this can also be interpreted as the atomic lifetime getting shorter and shorter. In this case, any photon absorbed by the TLS is immediately re-emitted, and thus the input mode shape is unchanged by the atom, again leading to $M\to 1$.

We observe that the above arguments, for both limiting cases, still hold for \textit{any} (square-integrable) temporal mode profile which is uniquely parametrized by a characteristic width. That is, for such types of modes, we still expect that $M\to 1$, when $\gamma/\Gamma \to  0$ and $\gamma/\Gamma \to  \infty$, for any number of photons.

The plots in Fig.~\ref{fig:number_and_singular_values} also reveal a remarkable feature. The number $M$ exhibits two local maxima located at $\gamma_1\approx 0.15\;\Gamma$ and $\gamma_2\approx 5.52\;\Gamma$, with values $M_1\approx 5.15$ and $M_2\approx 2.01$. It so happens that the peaks of $M$ are located exactly --- up to our numerical precision --- at values of $\gamma$ for which $\lambda_2=-\lambda_1$. These values are $\gamma_1\approx 0.15\;\Gamma$ and $\gamma_2\approx 5.52\;\Gamma$. This alignment is, in principle, coincidental, and we do not expect it to happen in general, as we discuss further in Appendix ~\ref{appendix:explicit_expression_gaussian}.
\begin{figure*}
    \centering
    \includegraphics[width=\linewidth]{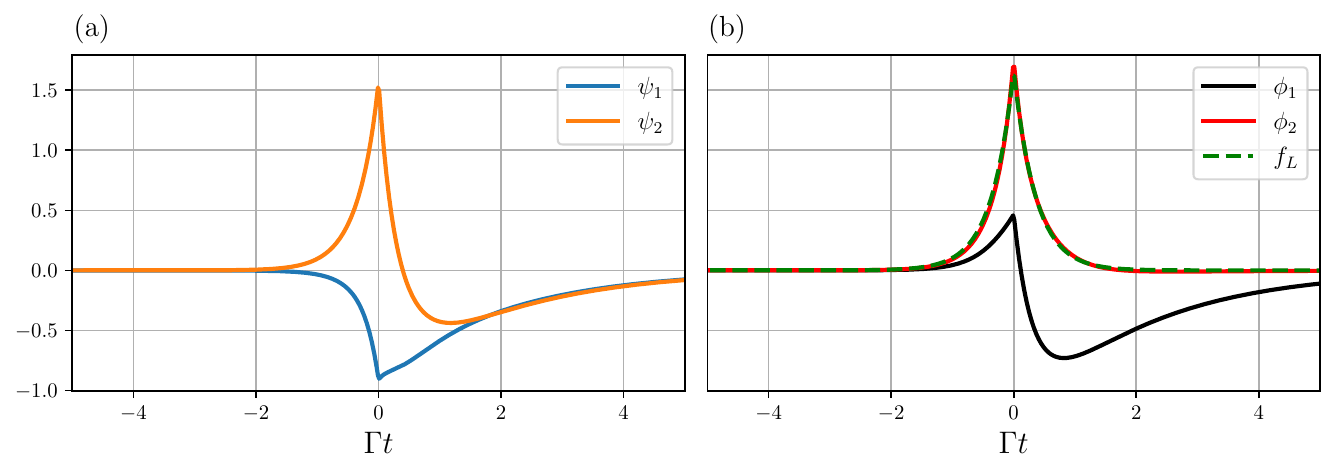}
    \caption{Different output modes for the Lorentzian input case [see Eq.~\eqref{eq:lorentzian_time}], choosing $\gamma=\gamma_2\approx5.55\; \Gamma$. In (a), we plot the first two principal modes $\psi_1$ and $\psi_2$, calculated numerically from Eq.~\eqref{eq:takagi_two_modes}. In (b), we depict the modes in the rotated basis, $\phi_1$ and $\phi_2$, which were obtained from Eqs.~\eqref{eq:rotated_basis_1} and~\eqref{eq:rotated_basis_2}, and compare them with the corresponding input mode $f_L$ [see Eq.~\eqref{eq:lorentzian_time}].}
    \label{fig:takagi_modes}
\end{figure*}

The second peak of $M$ is of special interest of this work. At the corresponding value of $\gamma$, not only the first two eigenvalues coincide in their moduli, but they are also much larger than all the other remaining eigenvalues. Thus at $\gamma_2$ the effective number of modes is very close to $2$ (with a deviation of $0.6\%$) and the diagonalization of the output state is well-approximated by only two modes, as
\begin{align}
    \label{eq:takagi_two_modes}
    \Psi_{\rm out}(t,t')\big|_{\gamma=\gamma_2}\approx\lambda[\psi_1(t)\psi_1(t')-\psi_2(t)\psi_2(t')],
\end{align}
with $\lambda=1/\sqrt 2$. We recognize in $\Psi_{\rm out}$ the wavefunction of a Hang-Ou-Mandel state, also known as two-photon NOON state. The principal modes $\psi_1$ and $\psi_2$ are plotted in Fig.~\ref{fig:takagi_modes}.

The form of Eq.~\eqref{eq:takagi_two_modes} further implies that one can define a rotated basis of modes, as
\begin{align}
    \label{eq:rotated_basis_1}
    \phi_1(t) &= \frac1{\sqrt 2}(\psi_1(t)+\psi_2(t))\\
    \label{eq:rotated_basis_2}
    \phi_2(t) &= \frac1{\sqrt 2}(\psi_1(t)-\psi_2(t)),
\end{align}
such that $\phi_1(t)$ and $\phi_2(t)$ decompose the wavefunction as
\begin{align}
    \label{eq:unsyn_two_modes}
    \Psi_{\rm out}(t,t')\big|_{\gamma=\gamma_2}\approx\lambda[\phi_1(t)\phi_2(t')+\phi_1(t')\phi_2(t)].
\end{align}

The modes $\phi_1(t)$ and $\phi_2(t)$ are also plotted in Fig.~\ref{fig:takagi_modes}. We observe on this figure that $\phi_2(t)\approx f_L(t)$ [see Eq.~\eqref{eq:lorentzian_time}], with the computed residual sum of squares amounting to $0.17\%$. Therefore, we conclude that, for $\gamma=\gamma_2$, the output state not only populates only two modes, but one of them is the same as the original mode.

We have so far written our analysis in terms of the output state wavefunction. The formulation of our results in Dirac notation is nevertheless straightforward. Let $\ket{n_{f}}=(n!)^{-1/2}(a^\dagger_f)^n\ket{0}$ denote the quantum state of $n$ photons in mode $f$ [see Eq.~\eqref{eq:mode_creation}]. Then for $\gamma=\gamma_2$, the scattering event is well-approximated by the following mapping 
\begin{align}
    \label{eq:hong_ou_mandel_state}
    \ket{2_{f_L}}\longrightarrow\frac1{\sqrt2}(\ket{2_{\psi_1},0}-\ket{0,2_{\psi_2}})=\ket{1_{\phi_1},1_{f_L}},
\end{align}
where we recover that a single-mode two-photon state is mapped into an entangled two-mode output
state, i.e. a two-photon NOON state.

Concerning the Wigner function, it is known that the state on the right-hand side of Eq.~\eqref{eq:hong_ou_mandel_state} has more Wigner negativity compared to the input state. Consequently, Eq.~\eqref{eq:hong_ou_mandel_state} depicts a deterministic generation, with a high degree of accuracy of states with non-Gaussian Wigner functions, which points towards a potential application in quantum computing.

\subsection{Gaussian mode}
\label{sec:other_modes}

We have also employed our approach to investigate the Gaussian temporal mode profile, which reads
\begin{align}
    \label{eq:gaussian_in}
    f_G(t)=\frac{\sqrt\sigma}{\pi^{1/4}}e^{-\sigma^2 t^2/2}.
\end{align}
The Gaussian input yields an output state which cannot be expressed in terms of elementary functions. Still, the same analysis as in Sec.~\ref{sec:output_mode_analysis} can be carried out, and our results for the Gaussian input mode reveal to be analogous to the ones obtained for the Lorentzian. For this reason, only the main findings are given in this section, whereas the details can be found in Appendix~\ref{appendix:explicit_expression_gaussian}.

By numerically diagonalizing the output state corresponding to the Gaussian input, we again computed the effective number of modes $M$ as a function of the width $\sigma$. We found that $M$ also exhibits two local maxima, $M\approx5.56$ for $\sigma\approx 0.09\;\Gamma$, and $M\approx 2$ for $\sigma\approx2.59\;\Gamma$. In particular, the second maximum differs from $2$ by only $0.35\%$, so that for $\sigma\approx 2.6\;\Gamma$ the output state admits the same decompositions as in Eqs.~\eqref{eq:takagi_two_modes} and~\eqref{eq:unsyn_two_modes}. Furthermore, a numerical comparison revealed that, just like with the Lorentzian case, $\phi_2(t)\approx f_G(t)$, this time with an error of $0.12\%$, thus providing an even better approximation compared to the Lorentzian case.

Finally these results can be condensed in a completely analogous diagram in Dirac notation as in Eq.~\eqref{eq:hong_ou_mandel_state}. Our results show that the latter is an even better approximation for the scattering event of the Gaussian, compared with the Lorentzian case.
These findings are consistent with what has been found in \cite{Lund2023}, where an alternative --- fully numerical --- method was used \cite{Kiilerich2019, Kiilerich2020}.

\section{Summary and outlook}
\label{sec:summary}

In this work we addressed the problem of the scattering of a multiphoton and multimode state of light by a two-level system in a chiral waveguide from a modal perspective. Importantly, we found an explicit way to describe the photonic output state solely in terms of the input modes. This description has a straightforward physical interpretation, due to its connection with the scattering matrix formalism \cite{Caneva2015}, and significantly simplifies the computations compared to other available methods \cite{Maffei2022, Kiilerich2019, Kiilerich2020}, especially for the case of $n>2$ photons.

From our method, we further investigated the modal decomposition of the output state in the two-photon scattering case, for two distinct input modes. Our results reveal that the input modes can be fine-tuned in such a way that the scattering event deterministically maps the two-photon single-mode incoming state to an entangled state of two photons populating two orthogonal modes, to a high degree of accuracy. While this property holds for both the Lorentzian and Gaussian profiles, we observe that the latter yields a better approximation than the former. This property suggests a potential ``mode optimization'' protocol, i.e., choosing the input mode such that the output state yields an even better approximation of the aforementioned entangled state. 

Finally, we expect our analytical expression for the output wavefunction to help investigate the modal content of generic output states, and in particular states with higher photon number. These states would be useful in understanding the scattering of experimentally relevant Gaussian input states --- such as squeezed states. The direct modal analysis of these scattered states can thus provide insights for deterministic generation of non-Gaussian states.

\section{Acknowledgements}

The authors thank D.E. Chang, A. Browaeys, A. Gavalda, I. Ferrier-Barbut, H. Le Jeannic, T. Copie and V. Shatokhin for valuable discussions. The authors gratefully acknowledge
support from the European Union under EIC Pathfinder Grant No. 101115420 (PANDA). This work was supported through the action France 2030 from Agence Nationale de la Recherche, projects NISQ2LSQ (ANR-22-PETQ-0006) and OQULUS (ANR-22-PETQ-0013). 

\appendix

\section{Normalization}
\label{appendix:normalization}

In this appendix we derive an explicit expression for the normalization factor $\mathcal N$ of the input state, appearing in Eq.~\eqref{eq:input_state}, and deduce Eq.~\eqref{eq:output_wavefunction}, regarding the wavefunction of the output state.

\subsection{Normalization of the input state}

We determine the factor $\mathcal N$ in Eq.~\eqref{eq:input_state} by imposing the input state $\ket{\psi_{\rm in}}$ to be normalized. That is, we compute
\begin{widetext}
\begin{align}
    \braket{\psi_{\rm in}}{\psi_{\rm in}} = \mathcal{N}^2 \int_{\mathbb R^n}\dd^n\vb t\int_{\mathbb R^n}\dd^n\vb u\;\prod_{j=1}^nf_j(u_j)f^*_j(t_j) \bra{0}a(t_1)\ldots a(t_n) a^\dag(u_1)\ldots a^\dag(u_n) \ket{0}=1.
\end{align}
\end{widetext}

One can then simplify the above expression by using the commutation relations between the ladder operators [see. Eq.~\eqref{eq:commutation_time}]. This manipulation yields a product of delta-functions \cite{Weinberg1995}, which, in turn, can be easily integrated over the $\{t_j\}$ variables. Finally, solving for $\mathcal N$, one finds
\begin{align}
    \label{eq:normalization}
    \mathcal{N} = \frac{1}{\sqrt{\sum\limits_{\sigma\in \mathfrak{S}_n}\prod\limits_{j=1}^n \langle f_{\sigma(j)}, f_j \rangle}},
\end{align}
where $\langle f, g \rangle$ is the scalar product between $f$ and $g$:
\begin{align}
    \langle f, g \rangle \coloneq\int_{\mathbb R} \dd u\;f^*(u) g(u).
\end{align}
We thus see that if all the $f_j$ are orthogonal, $\mathcal{N} = 1$, whereas if all the $f_j$ are equal --- i.e., $n$ photons in the same mode --- then $\mathcal{N} = 1/{\sqrt{n!}}$.

\subsection{Output state wavefunction}

The output state is obtained by applying the scattering matrix to the input state, as given by Eq.~\eqref{eq:application_scatter} in the main text. Using the definition of our input state, given by Eq.~\eqref{eq:input_state}, we have that
\begin{align}
    &\ket{\psi_{\rm out}} = \mathcal S\ket{\psi_{\rm in}}\nonumber\\
    &=\mathcal{N}\int_{\mathbb R^n}\dd^n\vb u \left[\prod_{j=1}^n f_j(u_j)\right]\mathcal Sa^\dagger(u_1)\ldots a^\dagger(u_n)\ket{0}.
\end{align}

Next, we insert in the above result, the identity operator in the form
\begin{align}
    \mathbb I=\frac 1{n!}\int_{\mathbb R^n}\dd^n\vb t\;a^\dag(t_1)\ldots a^\dag(t_n) \ketbra{0}a(t_n)\ldots a(t_t),
\end{align}
where the $1/n!$ factor accounts for the $n!$ permutations of the (complete and orthonormal) basis $a^\dag(t_1)\ldots a^\dag(t_n) \ket{0}\big|_{t_n\geq t_{n-1}\geq\ldots \geq t_1}$. 
We thus have
\begin{align}
    &\ket{\psi_{\rm out}}=\frac{\mathcal{N}}{n!}\times\nonumber\\&\int_{\mathbb R^n}\dd^n\vb t\int_{\mathbb R^n}\dd^n\vb u \left[\prod_{j=1}^n f_j(u_j)\right] S^{(n)}_{(\vb t, \vb u)}a^\dag(t_1)\ldots a^\dag(t_n) \ket{0},
\end{align}
where we employed the definition of the time domain representation of the scattering matrix [see Eq.~\eqref{eq:s_matrix_time}].

Finally, to recover the output state wavefunction, we simply use its definition, as given in Eq.~\eqref{eq:output_wave_definition}, and we find
\begin{equation}
    \Psi_{\rm out}(\vb t)=\frac{\mathcal{N}}{\sqrt{n!}}\int_{\mathbb R^n} \dd^n\vb u\; f_1(u_1)\ldots f_n(u_n) S^{(n)}_{(\vb t,\vb u)},
\end{equation}
where the factor $\mathcal{N}/\sqrt{n!}$ ensures that the normalization condition $\int_{\mathbb R^n} \dd^n\vb t\; |\Psi_{\rm out}(\vb t)|^2 = 1$ is fulfilled. We note that $\Psi_{\rm out}(\vb t)$ is always symmetric with respect to any permutation of $\vb t$ because $S^{(n)}_{(\vb{t},\vb{u})}$ is so.

\section{General formula for $n$-photon output state}
\label{appendix:derivation_general_formula}

In order to compute the output state wavefunction, we start from~\eqref{eq:output_wavefunction}, which we repeat below for clarity,
\begin{align}
    \label{eq:output_state_appendix}
    \Psi_{\rm out}(\vb t) = \frac{\mathcal N}{\sqrt{n!}}\int\limits_{\mathbb{R}^n}\dd^n\vb u f_1(u_1)\ldots f_n(u_n) S^{(n)}_{(\vb t, \vb u)}.
\end{align}
The whole calculation then amounts to plug the expression for $S^{(n)}_{(\vb t, \vb u)}$, given by Eq.~\eqref{eq:scattering_matrix}, in the above equation and integrate over the $\vb u$ variables.

However, Eq.~\eqref{eq:scattering_matrix} is valid only when $\vb u$ is ordered. We thus split the integral in Eq.~\eqref{eq:output_state_appendix} according to the ordering of $\vb u$, as
\begin{align}
    \label{eq:output_ordered_time}
    &\Psi_{\rm out}(\vb t) = \frac{\mathcal N}{\sqrt{n!}}\sum_{\sigma \in\mathfrak{S}_n}\int_{u_{\sigma(1)}\leq \ldots \leq u_{\sigma(n)}}\mkern-48mu\dd^n\vb u\;f_1(u_1)\ldots f_n(u_n) S^{(n)}_{(\vb t, \vb u)}\nonumber\\
    &=\frac{\mathcal N}{\sqrt{n!}} \sum_{\sigma \in\mathfrak{S}_n}\int_{u_{1}\leq \ldots \leq u_{n}}\mkern-36mu\dd^n\vb u\; f_{\sigma(1)}(u_{1})\ldots f_{\sigma(n)}(u_{n}) S^{(n)}_{(\vb t, \vb u)},   
\end{align}
where $\sigma$ denotes a given permutation in the symmetric group $\mathfrak S_n$ of $n$ elements. In the above, the second line is obtained after reordering the $\{f_j\}$ functions, relabeling the integration variables, and using the fact that $S^{(n)}_{(\vb t,\vb u)}$ is symmetric with respect to the permutations in $\vb u$.

With the result in the form of Eq.~\eqref{eq:output_ordered_time}, we are now able to employ expressions in Eqs.~\eqref{eq:scattering_matrix} and~\eqref{eq:connected_s}. We get
\begin{widetext}
\begin{align}
    \label{eq:derivation1}
    &\Psi_{\rm out}(\vb t)\nonumber\\ 
    =&\frac{\mathcal N}{\sqrt{n!}}\sum_{m=0}^n  \sum_{I,J\in \mathcal P(m)}\sum_{\sigma \in\mathfrak{S}_n}\int_{u_{1}\leq \ldots \leq u_{n}}\mkern-36mu \dd^n\vb u f_{\sigma(1)}(u_{1})\ldots f_{\sigma(n)}(u_{n}) \prod_{l=1}^m\left[-\Gamma e^{-\Gamma(t_{I_l} - u_{J_l})/2} \Theta(t_{I_l} \geq u_{J_l} \geq t_{I_{l-1}})\right]\cdot \prod_{l=1}^{n-m}\delta(t_{\Bar I_l} - u_{\Bar J_l})\nonumber \\
    =&\frac{\mathcal N}{\sqrt{n!}} \sum_{\sigma \in\mathfrak{S}_n}\sum_{m=0}^n  \sum_{I\in \mathcal P(m)}\sum_{J\in \mathcal P(m)}\int_{u_{1}\leq \ldots \leq u_{n}}\mkern-36mu  \dd^n\vb u  \prod_{l=1}^m\left[-\Gamma f_{\sigma(l)}(u_{J_l})e^{-\Gamma(t_{I_l} - u_{J_l})/2} \Theta(t_{I_l} \geq u_{J_l} \geq t_{I_{l-1}})\right]\cdot \prod_{l=1}^{n-m}f_{\sigma(l+m)}(u_{\Bar J_l})\delta(t_{\Bar I_l} - u_{\Bar J_l})\nonumber \\
    =&\frac{\mathcal N}{\sqrt{n!}} \sum_{\sigma \in\mathfrak{S}_n}\sum_{m=0}^n\sum_{I\in \mathcal P(m)} L_{\sigma,I},
\end{align}
where the second equality is obtained from the first one by relabeling the permutations $\sigma$ for each $J$. We also introduced the object $L_{\sigma,I}$, which encompasses all terms in the products, integrals and the summation over $I\in\mathcal P(m)$.

The heart of our result is the following
\begin{lemma}
\label{onlylemma}
    $$L_{\sigma,I} = \prod_{l=1}^m F_{\sigma(l)}(t_{I_l}, t_{I_{l-1}})\cdot \prod_{l=1}^{n-m} f_{\sigma(l+m)}(t_{\Bar I_l}).$$
\end{lemma}

\begin{proof}
From Eq.~\eqref{eq:derivation1}, $L_{\sigma,I}$ is defined as
\begin{align}
    \label{eq:definition_L}
    L_{\sigma,I} \coloneq \sum_{J\in \mathcal P(m)}\int_{u_{1}\leq \ldots \leq u_{n}}\mkern-36mu  \dd^n\vb u  \prod_{l=1}^m\left[-\Gamma f_{\sigma(l)}(u_{J_l})e^{-\Gamma(t_{I_l} - u_{J_l})/2} \Theta(t_{I_l} \geq u_{J_l} \geq t_{I_{l-1}})\right]\cdot \prod_{l=1}^{n-m}f_{\sigma(l+m)}(u_{\Bar J_l})\delta(t_{\Bar I_l} - u_{\Bar J_l}).
\end{align}
The main idea of the proof is then  to suitably relabel the indices in the expression of $L_{\sigma, I}$ in order to recover an integral over the whole $\mathbb R^n$ space. From that, the full integral over $\vb u$ can then be separated as a product of individual integrals over $u_l$, and the definition of Eq.~\eqref{eq:convoluted_mode} can be employed. 

We start by observing that the index map $l\to J_l$ acts as a permutation of the integration variables $\{u_l\}$. To make this parallel more explicit, we represent the latter map with the following function: for $J\in \mathcal P(m)$, let
\begin{align}
    \mu_J : l\in\{1,\ldots,n\} \mapsto \mu_J(l) =
    \begin{cases}
        J_l &\text{ if }l\leq m \\
        \Bar J_{l-m} &\text{ else.}
    \end{cases}
\end{align}
which is well-defined, since we took the convention that the elements of $J$ and $\bar J$ are in ascending order. Hence, we rewrite $L_{\sigma, I}$ in terms of the newly defined notation by replacing $J_l\to\mu_J(l)$ and $\Bar J_l\to \mu_J(l+m)$ in the indices of the $\vb u$ variables. We get
\begin{align}
    \label{eq:L}
    L_{\sigma,I} = \sum_{J\in \mathcal P(m)}\int_{u_{1}\leq \ldots \leq u_{n}}\mkern-36mu  \dd^n\vb u\;\;  \prod_{l=1}^m&\left[-\Gamma f_{\sigma(l)}(u_{\mu_J(l)})e^{-\Gamma(t_{I_l} - u_{\mu_J(l)})/2} \Theta(t_{I_l} \geq u_{\mu_J(l)} \geq t_{I_{l-1}})\right]\nonumber\\\times \prod_{l=1}^{n-m}&f_{\sigma(l+m)}(u_{\mu_J(l+m)})\delta(t_{\Bar I_l} - u_{\mu_J(l+m)})
\end{align}

The convenience of employing the above index relabeling relies on the fact that, when $J$ spans over $\mathcal P(m)$, $\mu_J$ spans over all the permutations of $\{1,\ldots,n\}$ whose first $m$ elements and last $n-m$ elements are in ascending order, individually. Crucially, for any permutation \textit{other} than the one specified by $\mu_J$, the integral over the domain $u_{1}\leq \ldots \leq u_{n}$ \textit{vanishes}, due to the ordering of $\vb t$ in the Dirac and Heaviside functions in the integrand in Eq.~\eqref{eq:L}. As a result, we can replace the sum over $J$ by a sum over all permutations $\mu_J\in\mathfrak S_m$, without affecting the final result, that is 
\begin{align}
    L_{\sigma,I}=\sum_{\mu \in \mathfrak{S}_n}\int_{u_{1}\leq \ldots \leq u_{n}}\mkern-36mu  \dd^n\vb u  \prod_{l=1}^m\left[-\Gamma f_{\sigma(l)}(u_{\mu(l)})e^{-\Gamma(t_{I_l} - u_{\mu(l)})/2} \Theta(t_{I_l} \geq u_{\mu(l)} \geq t_{I_{l-1}})\right]\cdot \prod_{l=1}^{n-m}f_{\sigma(l+m)}(u_{\mu(l+m)})\delta(t_{\Bar I_l} - u_{\mu(l+m)}).
\end{align}

Finally, with the above replacement, the sum over $\mu\in\mathfrak S_n$ has the effect of summing over all ordered possible domains $u_{1}\leq \ldots \leq u_{n}$, thus recovering an integral over the whole $\mathbb R^n$ space. We thus have that
\begin{align*}
    L_{\sigma, I}=& \int_{\mathbb{R}^n} \dd^n\vb u  \prod_{l=1}^m\left[-\Gamma f_{\sigma(l)}(u_{l})e^{-\Gamma(t_{I_l} - u_{l})/2} \Theta(t_{I_l} \geq u_{l} \geq t_{I_{l-1}})\right]\cdot \prod_{l=1}^{n-m}f_{\sigma(l+m)}(u_{l+m})\delta(t_{\Bar I_l} - u_{l+m})\\
    =& \prod_{l=1}^m \int_{\mathbb{R}} \dd u \left[-\Gamma f_{\sigma(l)}(u)e^{-\Gamma(t_{I_l} - u)/2} \Theta(t_{I_l} \geq u \geq t_{I_{l-1}})\right]\cdot \prod_{l=1}^{n-m}\int_{\mathbb{R}^n}\dd u f_{\sigma(l+m)}(u)\delta(t_{\Bar I_l} - u)\\
    =& \prod_{l=1}^m F_{\sigma(l)}(t_{I_l}, t_{I_{l-1}})\cdot \prod_{l=1}^{n-m} f_{\sigma(l+m)}(t_{\Bar I_l}),\\
\end{align*}
which concludes the proof.
\end{proof}

Let us now return to Eq.~\eqref{eq:derivation1}. Using Lemma \ref{onlylemma}, we arrive at the final result, namely
\begin{equation}
    \Psi_{\rm out}(\vb t) = \frac{\mathcal N}{\sqrt{n!}}\sum_{\sigma \in\mathfrak{S}_n}\sum_{m=0}^n  \sum_{I\in \mathcal P(m)}\prod_{l=1}^m F_{\sigma(l)}(t_{I_l}, t_{I_{l-1}})\cdot \prod_{l=1}^{n-m} f_{\sigma(l+m)}(t_{\Bar I_l}).
\end{equation}
In particular, if we assume that $f_1 = \ldots = f_n = f$, then the expression simplifies even more, as
\begin{equation}
    \Psi_{\rm out}(\vb t) = \sum_{m=0}^n  \sum_{I\in \mathcal P(m)}\prod_{l=1}^m F(t_{I_l}, t_{I_{l-1}})\cdot \prod_{l=1}^{n-m} f(t_{\Bar I_l}),
\end{equation}
which is Eq.~\eqref{eq:output_n_photons} in the main text.

\end{widetext}

\section{Explicit expression for the output state with a Lorentzian input}
\label{appendix:explicit_expression_lorentzian}

An explicit expression for the output wavefunction of two photons in mode $f_L$ is thus obtained by plugging Eq.~\eqref{eq:lorentzian_time} into Eq.~\eqref{eq:output_two_photon}. One finds
\begin{widetext}
\begin{align}
    \label{eq:output_lorentzian}
    \Psi_{\rm out}(t,t') =& \frac{\gamma}2 e^{-\gamma(|t|+|t'|)/2}\nonumber\\
    &\times
    \begin{cases}
        \displaystyle 1+\frac{4\Gamma\gamma}{(\gamma-\Gamma)^2}\left[1-e^{(\gamma-\Gamma)t'/2}-\frac{(\gamma-3\Gamma)}{(\gamma+\Gamma)}e^{(\gamma-\Gamma)t/2}-\frac\Gamma\gamma e^{(\gamma-\Gamma)(t-t')/2}\right] &\text{if } t>t'>0,\\
        \displaystyle 1-\frac{4\Gamma\gamma}{(\gamma+\Gamma)^2}\left[e^{(\gamma-\Gamma)t/2}+\frac\Gamma\gamma e^{((\gamma+\Gamma)t'+(\gamma-\Gamma)t)/2}\right] &\text{if }  t>0>t',\\
        \displaystyle 1-\frac{4\Gamma\gamma}{(\gamma+\Gamma)^2}\left[1+\frac\Gamma\gamma e^{(\gamma+\Gamma)(t'-t)/2}\right] &\text{if } 0>t>t',
    \end{cases}
\end{align}
\end{widetext}
and, for the remaining regions, we just have to swap $t\leftrightarrow t'$. We observe that, despite the presence of a $\gamma-\Gamma$ in the denominator in the first line of Eq.~\eqref{eq:output_lorentzian}, the wavefunction is well-defined in the limit $\gamma\to\Gamma$.

\section{Output state expression and analysis details with a Gaussian input}
\label{appendix:explicit_expression_gaussian}
Similarly, in the case of the Gaussian mode, we plug Eq.~\eqref{eq:gaussian_in} into Eq.~\eqref{eq:output_two_photon}. The obtained output wavefunction reads
\begin{align}
    \label{eq:state_out_gaussian}
    &\Psi_{\rm out}(t,t')=f_G(t)f_G(t')\nonumber\\\times&\left[1-\frac{\Gamma}{\gamma}\sqrt{\frac{\pi}2}(\mathrm{erfcx}(z)+\mathrm{erfcx}(z'))\right.\nonumber\\&\left.+\frac{\Gamma^2}{\gamma^2}\frac{\pi}2\left(\mathrm{erfcx}(z')-e^{z^2-{z'}^2}\mathrm{erfcx}(z)\right)\mathrm{erfcx}(z)\right],\;\mathrm{if}\;t\geq t',
\end{align}
where $z=(\Gamma-2\sigma^2t)/(2\sigma\sqrt 2)$, with an analogous definition for $z'$, and $\mathrm{erfcx}(x)=e^{x^2}(1-\mathrm{erf}(x))$, where $\mathrm{erf}(x)$ is the error function. In the case of $t'>t$, one has to swap the time arguments in Eq.~\eqref{eq:state_out_gaussian}.

\begin{figure}
    \centering
    \includegraphics[width=\linewidth]{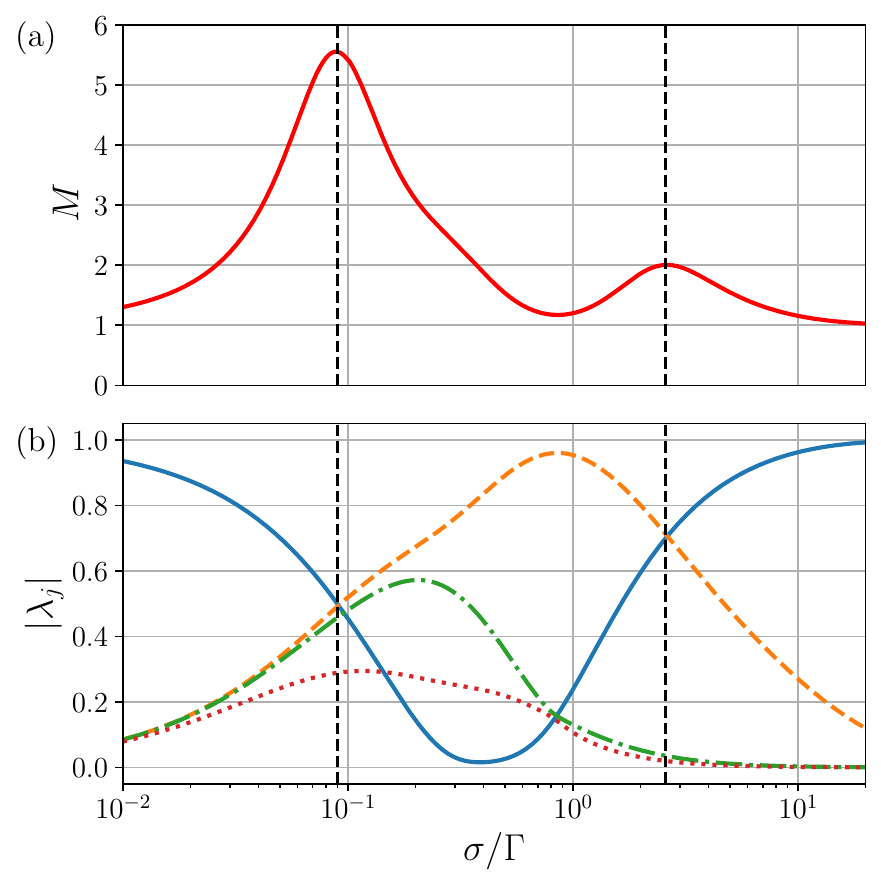}
    \caption{Analysis of the modal content of the output state for the Gaussian input. The quantities plotted in (a) and (b) are analogous to the ones in Fig.~\eqref{fig:number_and_singular_values}, with the difference that $\sigma$ is the Gaussian's width, defined by Eq \eqref{eq:gaussian_in}. The dashed vertical black lines highlight the local maxima of $M$: $M\approx 5.56$ at $\sigma\approx 0.09\;\Gamma$, and $M\approx 2$ at $\sigma\approx2.59\;\Gamma$.}
    \label{fig:number_and_sv_gaussian}
\end{figure}

Finally, we show in Fig.~\ref{fig:number_and_sv_gaussian} the output mode analysis for the Gaussian input, i.e., the plot of the mode number $M$ [see Eq.~\eqref{eq:mode_number}] and the first four eigenvalues $\lambda_j$ as a function of the Gaussian width $\sigma$ [cf. Fig.~\eqref{fig:number_and_singular_values}].

By numerically diagonalizing the output state corresponding to the Gaussian input, we computed the eigenvalues and the average number of modes $M$, both as a function of the width $\sigma$. We again found that $M$ has two local maxima, with the values $M_1\approx 5.56$ and $M_2\approx 2$, where the second deviates from $2$ by $0.35\%$. These peaks, in turn, are located at $\sigma_1\approx0.09\;\Gamma$ and $\sigma_2\approx 2.59\;\Gamma$. Regarding the eigenvalues, the Gaussian case is slightly different from the Lorentzian case: the two values of $\sigma$ for which $\lambda_2=-\lambda_1$ are $\sigma'_1=\sigma_1$ (up to our numerical precision) and $\sigma_2=2.68\;\Gamma\neq\sigma'_2$.

Despite this discrepancy, we still have that $|\sigma'_2-\sigma_2|$ is much smaller than the characteristic width of the peak associated with $M_2$. This property, together with the fact that, at $\sigma_2$, $|\lambda_1|=|\lambda_2|$ are much larger than all the other eigenvalues, still allows us to well-approximate the output state by the decompositions in Eq.~\eqref{eq:takagi_two_modes} and~\eqref{eq:unsyn_two_modes}. Similarly as before, we observed that the mode defined by Eq.~\eqref{eq:rotated_basis_2} is very close to the input mode, $\phi_2(t)\approx G(t)$, this time with a residual sum of squares of $0.12\%$, thus providing an even better approximation compared to the Lorentzian case.

\bibliographystyle{unsrt}
\bibliography{Paper-files/paper_bibliography}

\end{document}